\documentclass[a4paper]{article}
\usepackage{fullpage,amsmath,amsthm,amssymb,cleveref, graphicx}

\usepackage{thm-restate,url}
\newtheorem{problem}{Problem}
\newtheorem*{hypothesis}{Hypothesis}
\newtheorem{definition}{Definition}
\newtheorem{theorem}{Theorem}
\newtheorem{lemma}{Lemma}
\newtheorem{claim}{Claim}
\newtheorem{corollary}{Corollary}

\newcommand{\ignore}[1]{}

\newcommand{\paths}{\mathcal{\Pi}}
\newcommand{\w}{\mathcal{W}}
\newcommand{\C}{\mathfrak{C}}
\newcommand{\G}{\mathcal{G}}
\newcommand{\R}{\mathbb{R}}
\newcommand{\range}{\mathcal{R}}

\newcommand{\reach}{\mathcal{\, G_{\ge} }}
\newcommand{\subcat}{\mathcal{\, G_{\le} }}

\title{Hierarchical Categories in Colored Searching} 

\author{Peyman Afshani\thanks{Aarhus University, Denmark
    \texttt{peyman@cs.au.dk}. Supported by Independent Research Fund Denmark (DFF) grant ID DFF$-$7014$-$00404.}\and Rasmus Killmann\thanks{Aarhus
  University, Denmark \texttt{killmann@cs.au.dk}. Supported by Independent Research Fund Denmark (DFF) grant ID DFF$-$7014$-$00404.} \and Kasper Green
Larsen\thanks{Aarhus University, Denmark
  \texttt{larsen@cs.au.dk}. Supported by Independent Research Fund Denmark (DFF) Sapere Aude Research Leader grant No 9064-00068B.}
}

\begin{document}
\date{}
\maketitle

\thispagestyle{empty}

\begin{abstract}
    In colored range counting (CRC), the input is a set of points where each point is assigned
    a ``color'' (or a ``category'') and the goal is to store them in a data structure such
    that the number of distinct categories inside a given query range can be counted
    efficiently. 
    CRC has strong motivations as it allows data structure to deal with categorical data. 
    
    However, colors (i.e., the categories) in the CRC problem do not have any internal structure, whereas this is not the case for
    many datasets in practice where hierarchical categories exists or where
    a single input belongs to multiple categories. 
    Motivated by these,  we consider variants of the problem where such structures
    can be represented. 
    We define two variants of the problem called hierarchical range counting (HCC) and
    sub-category colored range counting (SCRC) and consider hierarchical structures
    that can either be a DAG or a tree. 
    We show that the two problems  on some special trees 
    are in fact equivalent to other well-known problems
    in the literature. 
    Based on these, we also give efficient data structures when the
    underlying hierarchy can be represented as a tree. 
    We show a conditional lower bound for the general case when the existing hierarchy can be any DAG, 
    through reductions from the orthogonal vectors problem. 
\end{abstract}

\section{Introduction}\label{sec:intro}
Range searching is a broad area of computational geometry where the goal is to store a given set $P$ of
input points in a data structure such that one can efficiently report (\textit{range reporting}) or count 
(\textit{range counting}) the points inside a geometric query region $\range$.
Sometimes, each point $p_i \in P$ is associated with a weight $w_i \in \R$ and the goal could be to find the sum
of the weights in $\range$, or the maximum weight in $P\cap \range$
(\textit{range max} problem). 
This is a very broad area of research and there are many well-studied variants. 
See a recent excellent survey by Agarwal for more information~\cite{Agarwal.survey16}.

\textit{Colored} (or \textit{categorical}) range searching is an important variant where
each input point is associated with a \textit{category} which conceptually is represented as a color;
the goal of the query is then to report or count the number of distinct colors inside the query range. 
Colored range searching has strong motivations since colors allow us to represent 
\textit{nominal attributes} such as brand, manufacturer and so on and in practice a data set often 
contains a mix of nominal and ordinal attributes. 

Colored range searching was introduced in 1993 by Janardan and Lopez~\cite{colored} and it 
has received considerable attention since then.
However, classical colored range searching only models a ``flat'' categorical structure where categories
have no inherent structure. 
We consider variants of colored range counting to model such structures. 
We show that looking at colored range searching from this angle creates a number of interesting
questions with non-trivial connections to other already existing problems.

\subsection{Problem Definitions and Motivations}
We begin by formally defining colored range counting. 

\begin{problem}[colored range counting]\label{prob:crc}
  Given an input set $P$ of $n$ points in $\R^d$, a set $C$ of colors (i.e., categories) and
  a function $\C: P \rightarrow C$,  store them in a data structure
  such that given a query range $\range \subset \R^d$, it can count the number of distinct colors in $\range$,
  i.e., the value  $|C_\range|$ where $C_\range = \left\{
  \C(p) \mid p \in P \cap \range \right\}$.

  In the weighted version, input also includes a weight function $\w: C \rightarrow \R$
  and the output of the query is the weighted sum of the distinct colors in $\range$,
  i.e., the value $\sum_{c \in C_\range} \w(c)$.
\end{problem}

Colored range searching assumes that the colors are completely unstructured and that each
point receives exactly one color. 
However, these assumptions can be inadequate as hierarchical categories are quite common. 
For example, biological classification of living organisms are done through a tree where inclusion
in a category implies inclusion in all the ancestor categories. 
In fact, similar phenomena happen with respect to most notions of classification (e.g., 
    classification of industries). 
In other scenarios, a point may have multiple categories (e.g., a car can have a ``brand'', a ``color'',
``fuel type'' and so on). While it is possible to view the set of categories assigned to a point as one single
category, doing so ignores a lot of structure. For example, ``a blue diesel car'' is both a ``blue car'' and
also a ``diesel car'' but by considering ``blue diesel'' as a single category, this information is lost. 
We believe it is worthwhile to study the notion of structures within categories from a theoretical point of view. 
We are not in fact the first in trying to do so and we will shortly discuss some of the previous attempts.

A natural way to represent hierarchical categories is to assume that vertices of a DAG $\G$ represent the set of categories
where an edge $\vec{e}=(u,v)$ from (a category) $u$ to (a category) $v$ means that 
$u$ is a sub-category of $v$. 
We call $\G$ a \textit{category DAG} (or a \textit{category tree} if it is a tree). 
For a vertex $v \in \G$, 
we define $\subcat(v)$ as the subset of vertices of $\G$ that can reach $v$ (i.e., ``sub-categories'' of $v$),
$\reach(v)$ as the subset of vertices of $\G$ that $v$ can reach (i.e., ``super-categories'' of $v$).
Similarly, for a subset $H \subset V(G)$ we define $\reach(H) = \cup_{v \in H}\reach(v)$, and 
$\subcat(H) = \cup_{v \in H}\subcat(v)$.

Category trees
allows us to represent hierarchical categories.
Category DAGs allow us to capture cases where points can have multiple categories. 
Consider the car example again.
We can define a category DAG $\G$ where a vertex $u\in \G$ represents the compound category $\left\{\mbox{diesel}, \mbox{blue}  \right\}$
with edges to vertices $d$ and $b$ that represent ``diesel'' and ``blue'' categories respectively.
Thus, the set $\subcat(d)$  represents all the ``diesel'' cars and it includes the category $u$, the ``blue diesel'' category,
and similarly, the set $\subcat(b)$ represents all the ``blue'' cars which also includes the category $u$, the ``blue diesel'' category.
Likewise, $\reach(u)$ includes both $b$ and $d$ since a ``blue diesel'' car is both a ``blue car'' and a ``diesel car''.

We revisit colored range searching problems, using concepts of category DAGs or trees.

\begin{problem}[sub-category range counting (SCRC)]\label{node_problem}
  Consider an input point set $P \subset \R^d$ of $n$ points, a DAG $\G$ with $O(n)$ edges, and a function $\C: P \rightarrow \G$.
  The goal is to store the input in a data structure, such that given a query that consists of a query range $\range \subset \R^d$ and a 
  query vertex $v_q \in \G$ it can output $|C_\range \cap \subcat(v_q)|$ where $C_\range = \left\{
  \C(p) \mid p \in P \cap \range \right\}$.
\end{problem}

\begin{problem}[hierarchical color counting (HCC)]\label{ancestor_problem}
  Consider an input point set $P \subset \R^d$ of $n$ points, a DAG $\G$ with $O(n)$ edges, and a function $\C: P \rightarrow \G$.
  The goal is to store the input in a data structure, such that given a query range $\range \subset \R^d$ one can
  output $|G_\range|$ where $G_\range$ is the set of colors in $\range$, i.e., 
  $G_\range= \bigcup_{p\in \range \cap P} \reach(\C(p))$.

  In the weighted version of the problem, each vertex $v$ (i.e., category) of $\G$ is associated with a weight
  $w(v)$ and given the query $\range$, the goal is to compute $\sum_{v \in G_\range} w(v)$ instead.
\end{problem}

\subparagraph{Related problems.} Very recently, there have been other attempts to consider the structure of ``colors'' within the computational
geometry community, e.g., He and Kazi~\cite{Meng.cpm21} consider a problem very similar to SCRC
on a tree $\G$; the only difference is that instead of $|C_\range \cap \subcat(v)|$, the query outputs
$|C_\range \cap \pi(v,w)|$ where $\pi(v,w)$ is a path between two query vertices $v, w \in \G$.
There are also other variants available. See~\cite{Meng.cpm21} for further references. 

\subsection{Previous and Other Related Results}
The study of colored range counting and its variants began in 1993~\cite{colored} and since then it has 
received considerable attention. See the survey on colored range searching and its variants~\cite{GuJaSm}. 
The problem has at least three main variants: color range counting (CRC),
color range reporting, and ``type 2'' color range counting (for every distinct
color, count how many times it appears). 
Here, we only review colored range counting results.

In 1D, one can solve the CRC problem using $O(n)$ space and $O(\log n)$ query time by an elegant and simple
transformation~\cite{Coloured_reporting_counting_3sided} which turns the problem into the unweighted 
2D orthogonal range counting problem which can be solved within said bounds~\cite{Chazelle.functional}. 
Interestingly, it is also possible to show an equivalence between the two problems~\cite{Kasper.Freek13}. 
The problem, however, is difficult for 2D and beyond. 
Kaplan et al.~\cite{Kaplanetal11} showed that answering $m$ queries on a set of $n$ points requires
$\Omega(n^{\omega/2-o(1)})$ time where $\omega$ is the boolean matrix multiplication exponent. 
Under some assumptions (e.g., the boolean matrix multiplication conjectures), this shows that
$P(n) + mQ(n) \ge n^{3/2-o(1)}$ where $P(\cdot)$ and $Q(\cdot)$ are the preprocessing time and the
query time of any data structure that solves the 2D CRC problem.

Note that the equivalence between 1D CRC and 2D range counting also applies
to the weighted case of both problems, however,
the status of the weighted 2D range counting is still unresolved.
It can be solved with $O(n \log n/\log\log n)$ space and $O(\log n/\log\log n)$ query time~\cite{Jaja.ISAAC04} but it is not known
if both space and query time can be improved simultaneously (it is possible to improve one at the expense of the other).
The only available lower bound is a query time lower bound of
$\Omega(\log n / \log\log n)$~\cite{patrascu08structures}.

Some other interesting problems related to our results are defined below. 

\begin{definition}
    The following problems are defined for an input that consists of a 
    set $P \subset \R^d$ of $n$ points.
    The goal is to build a data structure to answer the following queries. 
    \begin{itemize}
        \item (orthogonal range counting) Given a query axis-aligned rectangle $\range$,
            count the number of points in $\range$.
            In the \textit{weighted} version, the points have weights and the goal is to compute
            the sum of the weights in the query. 
        \item (dominance range counting) This is a special case of orthogonal range counting where the query
            rectangle has the form $(-\infty, q_1] \times \cdots (-\infty, q_d]$ which is also known as a 
            \textit{dominance range}. Orthogonal range counting and dominance range counting are known to be equivalent
            if subtraction of weights are allowed (e.g., integer weights).
        \item (3-sided color counting). The input is in the plane ($d=2$) and each point is assigned
            a color from a set $C$ and the query is a 3-sided range in the form of $\range = [q_\ell, q_r] \times (-\infty, q_t]$ and the goal is
            to count the number of colors in $\range$. 
            In the weighted version, the colors have weights and the goal is to compute the sum
            of the weights of the colors. 
        \item (3-sided distinct coordinate counting) This is a special case of 3-sided color counting where
            an input point $(x_i,y_i)$ has color $y_i$; in other words, given the query $\range = [q_\ell, q_r] \times (-\infty, q_t]$,
            we would like to count the number of distinct $Y$-coordinates inside it. 
        \item (range max) Given a weight function $\w: P \rightarrow \R$ as part of the input, at the query time we would like to find
            the maximum weight inside a given query range $\range$.
        \item (sum-max color counting) This a combination of range max and color counting queries. 
            Assume the points in $P$ have been assigned colors from a set $C$
            and assume we have a weight function
            $\w: P \rightarrow \R$  on the points. Given a query range $\range$, we would like to compute the 
            output $\sum_{c \in C}X_c(\range)$ where $X_c(\range)$ is the maximum weight of a point with color $c$
            inside $\range$; if no point of color $c$ exists in $\range$, then $X_c(\range) = 0$.

            A sum-max color counting query includes a number of the above problems as its special case:
            If all points have the same weight, then it reduces to a color counting query. 
            If all points have the same color, then it reduces to a range max query. 
            If all points have distinct colors, then it reduces to a weighted range counting query. 
    \end{itemize}
\end{definition}

\subsection{Our Results}
Clearly, sub-category range counting (SCRC) is at least as hard as CRC. 
We also observe that hierarchical color counting (HCC) is also at least as hard.
Thus, getting efficient results for $d \ge 2$ seems hopeless. 
Consequently, we focus on the 1D problem but for two different important DAG's: when $\G$ is
a tree and also when $\G$ is an arbitrary (sparse) DAG. 
Our main results are the following. 

For the SCRC problem, 
first, we observe that the following problems are equivalent:
\begin{enumerate}
    \item SCRC when $\G$ is a single path on a one-dimensional input $P$.
    \item 3-sided distinct coordinate counting (for a planar point set $P$).
    \item 3-sided color counting (for a planar point set $P$).
    \item 3D dominance color counting (for a 3D point set $P$).
    \item 3D dominance counting (for a 3D point set $P$).
\end{enumerate}
We start by observing that (1) and (2) are equivalent. 
It is also clear that (2) reduces to (3); the reduction from (3) to (4) is standard
by mapping a 2D input point $(x_i,y_i)$ to the 3D point $(-x_i,y_i,x_i)$ and the 3-sided
query range $\range = [q_\ell, q_r] \times (-\infty, q_t]$ to the 3D dominance
range $(-\infty,-q_\ell]\times (-\infty,q_t] \times   (-\infty,q_x]$.
The reduction from (4) to (5) was shown 
by Saladi~\cite{Saladi_2021}.
We complete the loop by observing that (5) in turn reduces to (2). 
Note that the weighted versions of the problems are also equivalent by following the same reductions. 
Next, we show that SCRC can be solved using $O(n \log^2 n/\log\log n)$ space
and with query time of $O(\log n / \log\log n)$ on any category tree $\G$;
our query time is optimal which follows from the above reductions and using known results~\cite{patrascu08structures}.

For the HCC problem on 
trees, we show that (weighted) HCC in 1D can be solved using a 2D (weighted) range counting data structure
on $O(n\log n)$ points.
For example, this yields a solution with $O(n \log^2 n/\log\log n)$ space and with $O(\log n/\log\log n)$ query time. 
Interestingly, we show that the following problems are in fact equivalent:
\begin{itemize}
  \item Unweighted HCC in 1D when $\G$ is a (generalized) caterpillar graph.
  \item Weighted 2D range counting with $\Theta(\log n)$ bit long integer weights.
  \item 1D Colored range sum-max.
\end{itemize}
These reductions are quite non-trivial and they show a surprising  equivalence between
an unweighted problem (HCC) and the weighted 2D  range counting.
This allows us to directly apply known lower bounds or barriers for 2D range counting.
First, there is an $\Omega(\log n /\log\log n)$ lower bound~\cite{patrascu08structures} for 2D range counting with near-linear
space ($O(n\log^{O(1)}n)$ space) and by the above reductions, the same
bound also applies to unweighted HCC in 1D.

When $\G$ can be any arbitrary sparse DAG, the problems become more complicated. 
By a reduction from the orthogonal vectors problem, we show that we must either have
$\Omega(n^{2-o(1)})$ preprocessing time or the query time must be almost linear $\Omega(n^{1-o(1)})$ and
this holds for both SCRC and HCC.
Surprisingly, for the HCC problem, we can build a data structure that has $O(\log n)$ query time
using $\tilde{O}(n^{3/2})$ space, even though the data structure requires $\tilde{O}(n^{2})$ preprocessing time. 
This is one of the rare instances where there is a polynomial gap between the space complexity and the
preprocessing time of a data structure.

\section{Technical preliminaries}\label{sec:tech}
A fundamental technique to decompose trees into paths with certain properties
is called the \textit{heavy path decomposition}. The technique was originally used as
part of the amortized analysis of the link/cut trees introduced by Sleator and
Tarjan~\cite{link_cut_tree} and later used in the data structure construction
for lowest common ancestor by Harel and Tarjan~\cite{lowest_common_ancestor}.
The decomposition is simple and gives us the following properties.

\begin{theorem}
    Given a tree T of size $O(n)$, we can partition the (vertices of the) tree
    into a set of paths
$\pi_1, \pi_2,\dots \pi_h$ such that on any root to leaf path in T, the number
of different paths encountered is $O(\log(n))$.
\end{theorem}

\ignore{
Another way of looking at the problem could be that we are only interested in a certain part of the graph. If we as an example wants to know the number of Italian restaurants in London on a category graph over Europe, we only want to query a specific part of the graph. We try to model this in the following problem definition.

\begin{problem}[Node Problem]\label{node_problem}
Given a colored point set in $\mathbb{R}^d$ and an underlying category graph $G$, where each point is associated with a leaf in $G$ (a node with out degree zero), store the points and $G$ in some datastructure. Given an orthogonal query range and a node $v$ in $G$ output the size of the set of categories within the query range which is also reachable from $v$.
\end{problem}
} 

We study the HCC and SCRC problem in one dimension.
First, we consider when $\G$ is a tree in Section \ref{ancestor} and then we consider the general case where $\G$ can be any (sparse) DAG in Section \ref{sec:DAG}.
The general case is more difficult to solve and we will show that through a reduction (a ``conditional lower bound'').
Our reduction relies on the Orthogonal Vectors conjecture which is implied by
the Strong Exponential Time Hypothesis (SETH)~\cite{Williams-2con}.

\begin{hypothesis}[Orthogonal vectors conjecture]
Given two sets $A$ and $B$ each containing $n$ boolean vectors of dimension
$d=\log^{O(1)}(n)$, deciding whether there exists two orthogonal vectors $a_i \in A$ and $b_j\in B$ 
cannot be done faster than $n^{2-o(1)}$ time.
\end{hypothesis}

Assuming this conjecture, many near optimal time lower bounds have been proven within
$P$, including Edit Distance, Longest Common Subsequence and Fréchet distance~\cite{edit_condi, lcs_condi, frechet_condi}.

Finally, we say that a binary tree $T$ is a generalized caterpillar tree if all the 
degree three vertices lie on the same path (a caterpillar tree is typically defined
as the legs having size 1).

\section{Hierarchical Color Counting on Trees}\label{ancestor}
In the appendix (Section~\ref{app:ashard}), we observe that HCC is at least as hard as
CRC, using a simple reduction. 
As a result, we focus on the 1D case. 
We start off by presenting a data structure to solve 1D HCC on a tree $\G$ and then show that 
unweighted HCC on generalized caterpillars is actually equivalent to weighted 2D dominance counting (up to constant factors). 
This allows us to conclude that HCC on generalized caterpillar graphs cannot be solved with $o(n\log n/\log\log n)$ space and
$o(\log n/\log\log n)$ query time, unless the state-of-the-art on weighted 2D dominance counting can be improved.

\subsection{A Data Structure}
We will now focus on the 1D HCC problem on trees, as described
in Problem~\ref{ancestor_problem}. 
Our main result is the following. 

\begin{theorem}\label{thm:h}
  The HCC problem on a tree (both weighted or unweighted) in $\mathbb{R}$ can be solved using
$S(n)=O(n\log^2(n)/\log \log n)$ space and $Q(n)=O(\log(n)/\log\log n)$ query time.
\end{theorem}

We prove the above theorem in two steps, using the following lemma. 
\begin{lemma}\label{lem:h}
    (i) The HCC problem in $\R$ can be reduced to
    $O(\log n)$ sub-problems of the sum-max problem in $\R$ on $n$ points each.
    (ii) A sum-max problem on $n$ points can be reduced to a (weighted) 2D orthogonal
    range counting problem on $n$ points.
\end{lemma}
\begin{proof}
    Our approach starts by looking at the underlying tree structure in $\G$.
    We split $\G$ into its heavy path decomposition.
    To prove part (i) of the lemma, we actually need to look at the specific details of the
    heavy-path decomposition which can be described as follows. 
    Start from the root  of $\G$ and follow a path to a leaf of $\G$  where at every
    node $u$ of $\G$, we always descend to a child of $u$ that has the largest
    subtree; this easily yields the property that after removal of $\pi$, $\G$ will be decomposed
    into a number of forests where each forest is at most half the size of $\G$.
    Then the heavy-path decomposition is built recursively, by recursing on every resulting forest.
    It is easily seen that the depth of the recursion is $O(\log n)$. 

    Let $\paths_i$ be the set of paths obtained at the $i$'th-level of the recursion. 
    An important observation here is that the paths in $\paths_i$ are ``independent'' meaning,
    no vertex $u$ of a path $\pi \in \paths_i$ is a descendent or ancestor of a vertex
    $v$ of a different path $\pi' \in \paths_i$.
    As a result, we claim it is sufficient to solve the HCC problem 
    on the paths $\paths_i$, for every $i=1, \cdots, O(\log n)$, and then sum up the
    $O(\log n)$ results; the set of paths in $\paths_i$ defines the $i$-th
    sub-problem and thus it remains to show how it can be reduced to a sum-max problem. 

    We now build an instance of the sum-max problem on $\paths_i$:
    we use the same input set $P$  but with different colors and also
    the points will receive weights, as follows. 
    For every path in $\paths_i$, we define a new color class, i.e., 
    for the set $C$ in the definition of the sum-max problem we have $|C| = |\paths_i|$.
    Let $c$ be the (original) 
    color of a point $p$ in the HCC problem (i.e., in graph $\G$). 
    Consider the position of the color $c$ in $\G$ and the path $\pi_c$  that connects
    it to the root of $\G$.
    Let $\pi_j \in \paths_i$ be the path that intersects $\pi_c$ (if there's no such path, $p$ is not
    stored in the $i$-th subproblem) on a vertex $v$. 
    The weight of $p$ will be a prefix sum of the weights in $\pi_j$:
    we start from the root of $\pi_j$ and add the weights all the way to $v$.
    Note that an unweighted HCC can be thought of as a weighted instance of HCC with weights one.
    By the definition of the sum-max problem, the answer to a sum-max query will yield
    the number of vertices (or the total weights of the vertices) of the paths
    in $\paths_i$ that need to be counted in
    the HCC problem. 
    This concludes the proof of part (i) of the lemma. 

    To prove part (ii), 
    we transform each point into an orthogonal range in 2D, inspired by the previous
    solutions to 1D color counting. 
    Consider an instance of a sum-max problem where 
    the $x$-coordinate of a point $p$ is $p_x$,   its color is 
    $c(p)$  and  its weight is $w(p)$. 
    For a point $p^{(i)}$, denote the first point of the same color and greater weight to its left (resp. right) 
    with $p^{(\ell)}$ (resp. $p^{(r)}$).
    Observe $p^{(i)}$ is only counted by a sum-max query $I$ if
    we have both $p^{(i)}_x \in I$ and $I\subset (p^{(\ell)},p^{(r)})$.
    Based on this observation, 
    we associate to $p^{(i)}$ the two dimensional region
    $(p^{(\ell)}_x, p^{(i)}_x]\times [p^{(i)}_x,
    p^{(r)}_x)$ (If $p^{(\ell)}$ or $p^{(r)}$ does not exist, make the region
    unbounded in the corresponding direction). We do this transformation
    for all points in our input. For a query range $I=[I_1,I_2]$ we map
    it to the point $q=(I_1,I_2)$; by our observation, $q$ precisely stabs the
    rectangles which correspond to the heaviest points of every color that lies inside $I$.

    Thus, we have reduced the problem to the following: our input is a set
    of axis-aligned rectangles where
    each rectangle is assigned a weight and given a query point $q=(q.x,q.y)$, the goal is to sum up 
    the weights of the rectangles that contain $q$, a.k.a, an instance of the rectangle stabbing problem. 
    By a simple known reduction, this problem reduces to 2D orthogonal range counting:
    simply turn an input rectangle $[x_1,y_1]\times [x_2,y_2]$ where $x_1<x_2, y_1<y_2$ and with weight $w$ 
    into four points: points $(x_1,y_1)$ and $(x_2,y_2)$ with weight $w$ and 
    points $(x_1,y_2)$ and $(x_2,y_1)$ with weight $-w$.
    Computing the answer to the dominance query with point $(q.x,q.y)$ will count $w$ only when
    $q$ is inside the rectangle as otherwise the weights $w$ and $-w$ will cancel each other out. 
\end{proof}

Theorem~\ref{thm:h} follows easily from Lemma~\ref{lem:h} since we only need
$O(\log n)$ sum-max data structures on $O(n)$ points; each reduces to weighted orthogonal
range counting and the final observation is that we can combine all of the $O(\log n)$
data structures in one 2D orthogonal range counting data structure on 
$O(n \log n)$ points.
Using known results, this can be solved with 
$O(n\log^2(n)/\log\log n)$ space and $O(\log(n)/\log\log n)$ query time~\cite{c88} although other trade-offs are also
possible. 
For example, by plugging in other known results for weighted 2D range counting, we can also obtain
$O(n\log n)$ space and $O(\log^{2+\varepsilon}n)$ query time, for any constant $\varepsilon >0$.

\subsection{Lower Bounds and Equivalence}
Here we will look at equivalent problems to the HCC problem in 1D.
Aside from showing that this problem has interesting and non-trivial connections to other problems,
the results in this section imply an $\Omega(\log n/\log\log n)$ query lower bound for our problem
which shows that the query time of our data structure from the previous section is optimal.

\begin{theorem}\label{thm:eq}
    The following problems defined on an input set $P$ of size $n$ are equivalent,
    up to a constant factor blow up in space and query time and potentially an
    additive term in the query time for answering predecessor queries. 
    \begin{itemize}
        \item {[P1]}: Unweighted HCC on  a
            generalized caterpillar of size $O(n)$ in 1D.
        \item {[P2]}: The sum-max problem with $O(\log n)$ bit long integer weights in 1D. 
        \item {[P3]}: 2D orthogonal range counting with $O(\log n)$ bits long integer weights.
    \end{itemize}
\end{theorem}
\begin{proof}
    The argument in the previous section shows that 
    P1 reduces to P2  since 
    in a generalized caterpillar, there are only two levels in the heavy-path 
    decomposition so there is no blow up of a $\log n$ factor in the space complexity.
    P2 in 1D reduces to P3 using standard techniques, using the same transformation from 
    1D color counting to 2D range counting~\cite{Coloured_reporting_counting_3sided}.
    The non-trivial direction is to reduce P3 to P1.
    We do this in a step-by-step fashion. 
    \begin{claim}\label{cl:1}
        P3 can be reduced to $O(1/\varepsilon)$ orthogonal range counting problems on $\varepsilon \log n$ bit-long integer weights, 
        for any constant $\varepsilon > 0$. 
    \end{claim}
    \begin{proof}
        Let $X =2^{\varepsilon \log n}$. 
        Given a weighted point set $P$ for P3, store the weights modulo $X$ in a structure
        for $\varepsilon \log n$ bit weights. 
        Now, we can strip away the $\varepsilon \log n$ least significant bits of the original 
        weights and repeat this process $1/\varepsilon$ times to prove the claim.
    \end{proof}

    \begin{claim}\label{cl:2}
        For any constant $s$,
        P3 can be reduced to $O(n^{1-2^{-s}})$ sub-problems of P3 on instances of size
        $O(n^{2^{-s}})$ where a query in the original problem can be reduced to
        $O(1)$ queries on some of the sub-problems. 
    \end{claim}
    \begin{proof}
      We adapt the grid method by Alstrup et al.~\cite{abr00}.
        Observe that as weights are integers, summing up the weights inside
        a query rectangle can be done using additions and subtractions of four dominance
        queries of the form $(-\infty, x_1] \times (-\infty, x_2]$.

        Build a $\sqrt{n} \times \sqrt{n}$ grid such that each row and column 
        contains $\sqrt{n}$ input points. 
        Then, use a $\sqrt{n} \times \sqrt{n}$ table $T$ to store partial sums, as follows:
        the cell $(i,j)$ of $T$ stores the sum of all the weights in grid cells $(i',j')$ with $i' \le i$ and $j'\le j$. 
        After storing $T$, we recurse on the set of points stored in each row as well as each column and stop the recursion 
        at depth $s$. At every sub-problem at depth $s$ of the recursion, we have $O(n^{2^{-s}})$ points left and they
        become the claimed sub-problems in the lemma. 

        To bound the number of sub-problems, observe that if a sub-problem at depth $i$
        has $m$ points, then it creates 
        $2\sqrt{m}$ problems (one for every row and column) involving $\sqrt{m}$ points each in depth $i+1$.
        By unrolling the recursion, we can see that 
        at depth $s$ of the recursion, we will have $2^s n^{1-2^{-s}}= O(n^{1-2^{-s}})$ sub-problems since
        $s$ is a constant, as claimed. 

        Now consider a query $q=(-\infty, q_x] \times (-\infty, q_y]$. 
        Observe that by using two predecessor queries, we can find the grid cells $g$ that contains the query point.
        Assume $g$ corresponds to the cell $(i,j)$  in the table $T$ and consider the cell
        $(i-1,j-1)$. 
        We have stored the sum of all the weights in the cells $(i',j')$ with $i' < i$ and $j' < j$. 
        This value gives us the total sum of the weights in the grid
        cells that are completely contained in $q$.
        Next, $q$ can be decomposed into two two queries, one in a row containing the
        $(q_x,q_y)$ point and another one in the column containing the same point. 
        Furthermore, the two queries can be made disjoint by having the cell $(i,j)$ included
        in only one of them. 
        These queries can then be answered recursively until we reach the $s$-th level
        of the recursion.
        Thus, in total we will need to answer $2^s = O(1)$ queries. 
    \end{proof}

    \begin{claim}
        After performing the reductions in Claims~\ref{cl:1} and \ref{cl:2}, P3 can be reduced to
        an instance of P2 where there are at most $n^{\varepsilon}$ colors and where the 
        maximum weight is at most $n^{\varepsilon}$, for any constant $\varepsilon > 0$.
    \end{claim}
    \begin{proof}
        Pick $s$ in Claim~\ref{cl:2} such that each sub-problem has at most 
        $0.5 n^\varepsilon$ points. 
        Consider one such sub-problem involving $m$ points. 
        We do a reduction inspired by Larsen and Walderveen~\cite{Kasper.Freek13}.
        Consider an input point $p=(x_i,y_i)$ with weight $w(p)$.
        $p$ will be mapped to two points $(-x_i)$ and $(y_i)$.
        They are first stored in a sum-max data structure with weight $w(p)$ and
        color $i$.
        They are also stored in a second sum-max data structure with weights $w(p)$
        but with different colors  of $2i$ and $2i+1$.
        
        Now, given a query range  $q=(-\infty, q_x] \times (-\infty, q_y]$, 
        we query both sum-max data structures with interval $[-q_x, q_y]$ and then
        subtract their results. 
        If the point $p$ is inside $q$, the first data structure counts $w(p)$ once
        but the second data structure counts them twice and thus their subtraction includes
        $w(p)$ only once. 
        If $p$ is not inside $q$, none of the data structures counts $w(p)$ or both counts $w(p)$ and those they cancel out in the subtraction. 

        Finally, note that the total number of colors is at most $2m \le n^\varepsilon$. 
    \end{proof}

    \begin{claim}
        An instance of P2 where there are at most $n^{\varepsilon}$ colors and where the 
        maximum weight is at most $n^{\varepsilon}$, for any constant $0.5 > \varepsilon > 0$,
        can be reduced to P1 on a generalized caterpillar of size $O(n)$.
    \end{claim}
    \begin{proof}
        We build a caterpillar graph $\G$ with a central path of length $n^\varepsilon$.
        Then, we attach a path of length $n^\varepsilon$ to every vertex on the central 
        path; call these attached paths, \textit{legs}. 
        The total size of the caterpillar graph is at most $n^{2\varepsilon} \le n$.

        Consider an input point $p$, with color $i$ and weight $w\le n^{\varepsilon}$. 
        In our HCC problem, we assign it a color that corresponds to the
        $w$-th vertex of the $i$-th leg. 
        Now, observe that given a query $I=[I_1,I_2]$ to the sum-max problem, 
        asking the same query on $\G$  will produce the answer to the sum-max problem: all 
        the vertices on a leg have the same color which is different from the color of all the
        other legs. 
        Furthermore, at every leg we simply need to find its lowest vertex that is contained
        in the query range which is equivalent to finding the point of maximum weight in the
        same color class. Counting the number of vertices on the central path is equivalent 
        to a range max query which is a special case of sum-max queries.
    \end{proof}

    Observe that the proof follows directly using the above claims. Note that at each step,
    we might blow up the query time and the space by a constant factor.
    Also, Claim~\ref{cl:2} requires a constant number of predecessor queries.
    Depending on the assumptions on the coordinates of the points this can take a varying
    amount of time but it is dominated by the actual cost of answering the range counting
    queries in any reasonable model of computation. 
\end{proof}

As a consequence of the above equivalence, we can get a number of conditional lower bounds for
HCC queries on trees. 
\begin{corollary}
    The barrier of $S(n)Q(n) = \Omega(n (\log n/\log\log n)^2)$ for weighted 2D range counting data structure also applies 
  to unweighted HCC queries, even for graphs as simple as generalized caterpillar graphs. 
  The query lower bound of $\Omega(\log n /\log\log n)$ also applies to the HCC problem. 
  Here $S(n)$ and $Q(n)$ refer to the space and query complexities.
\end{corollary}

Finally we remark that if the graph $\G$ is a path, then the HCC problem simply reduces to the range max queries
which do have more efficient solutions (e.g., with $O(n)$ space and $O(1)$ query time~\cite{fischer10rmq} plus
$O(1)$ predecessor queries).
And thus, caterpillar graphs are among the simplest graphs on which the above reduction is possible.

\section{General Hierarchical Color Counting Queries}\label{sec:DAG}
Now we shift our attention to the problem where the underlying graph $\G$ 
is a directed acyclic graph (DAG). 
This variant is clearly more complicated than the tree variant. 
However, we observe a very curious behavior, namely, the preprocessing bound is
much higher than the space complexity. 
We show a conditional lower bound on the preprocessing time using a reduction from the orthogonal vectors problem.

\subsection{A Reduction from Orthogonal Vectors}
We will reduce the Orthogonal Vectors problem to the 1D HCC problem.

\begin{theorem}\label{ov_HCC}
Assuming the Orthogonal Vectors conjecture, any solution to the 1D hierarchical color
counting problem on a DAG using $P(n)$ preprocessing time and $Q(n)$ query time, must obey
$P(n)+nQ(n)\geq n^{2-o(1)}$.
\end{theorem}

Note that in the HCC problem, a query time of $O(|\G|)$ is trivial by simple graph traversal methods. 
As a result, the above reduction shows that any non-trivial solution (besides $n^{o(1)}$ 
factor improvements) must have a large preprocessing time. 
\begin{proof}
Let $\eta = \frac{n}{\log^{c}(n)}$ for a large enough constant $c$.
We build an instance of HCC with $\eta$ points, and a DAG $\G$ with $O(\eta)$ vertices
but with $O(n)$ edges. 
We reduce the orthogonal vectors problem on $\eta$ vectors of dimension $\log^{c}n$ 
to this instance of HCC, notice that $n^{2-o(1)} = \eta^{2-o(1)}$.

Given two sets $A$ and $B$ of $\eta$ boolean vectors in $\{0,1\}^d$, we will construct the following DAG
$\G$. $\G$ will have three layers. 
For each vector in $A$ create a vertex (i.e., a category) in what we denote the first layer.
Now for each of the $d$ coordinates of the vectors create a vertex in the second
layer. Lastly create a vertex in the third layer for each vector in $B$. 

Create the following edges: For a vertex corresponding to vector $a_i$ in the
first layer create an outgoing edge to all coordinate vertices in the second
layer in which $a_i$ has a one at that corresponding coordinate. Then, for a
coordinate vertex in the second layer create an outgoing edge to all vectors in
the third layer where the corresponding vector in $B$ has a one at that
coordinate. This clearly takes $O(\eta d) = O(n)$ to construct (see~\cref{fig:ov}).

\begin{figure}[h!]
\centering{
\includegraphics[scale=0.5]{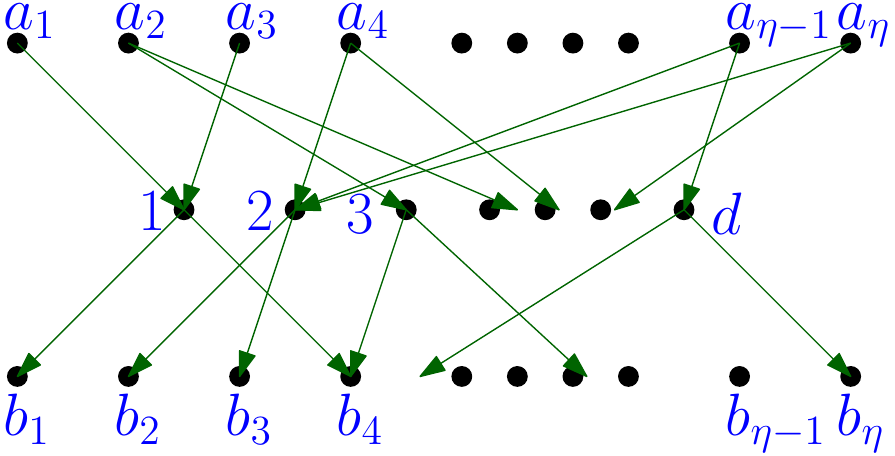}}
\caption{The underlying hierarchy DAG in the Orthogonal Vectors reduction}
\label{fig:ov}
\end{figure}

To figure out whether there exists a vector $a\in A$ and a vector $b\in B$ such
that $a$ and $b$ are orthogonal, we do the following. 
Create a point in $\R$ for each of the vertices in the third layer;
their locations do not matter as long as they are distinct. 
We use $n$ queries by simply querying each point individually and thus each query interval
has just one point inside it!
\begin{claim}
  Consider a HCC query that contains the point $p_i$ that corresponds to a vector
  $b_i \in B$. 
  There is no vector in $A$ that is orthogonal to $b_i$, if and only if 
  the output size
  is $|b_i|_1 +1 + \eta$ where $|b_i|_1$ is the number of ones in vector $b_i$.
\end{claim}
\begin{proof}
  First, let us consider the case when there is no vector in $A$ that is orthogonal
  to $b_i$. 
  Consider an arbitrary vector $a_j \in A$. 
  Since $a_j$ is not orthogonal to $b_i$, there exists a coordinate $k$ where both $b_i$ and $a_j$ have a 1 at that
  coordinate. 
  This implies that $a_j$ is connected to $b_i$ via the $k$-th vertex in the middle. 
  As a result, all vectors in $A$ are ancestors of $b_i$ and since $b_i$ is connected to $|b_i|_1$ vertices
  in the middle, the output of the HCC query will be as claimed. 
  
  The converse also follows by a similar argument. 
  If a vector $a_j\in A$ does not share a 1 coordinate with $b_i$, 
  then this corresponds to one of the vertices in the top layer not being counted by the
  query, hence the output is less than $|b_i|_1+1+\eta$. 
\end{proof}
Since one can easily store the values $|b_i|_1$ in $O(\eta)$ space, 
  we can solve the Orthogonal Vectors problem using $\eta $ queries on a solution for the HCC on the 
  aforementioned DAG. The DAG has size $O(\eta d) = O(n)$
  and thus we obtain the lower bound $P(n) + n Q(n)  \geq \eta^{2-o(1)} = n^{2-o(1)}$.
\end{proof}

\subsection{A Data Structure for General DAGs for HCC}
Despite the lower bound in the previous section, it is possible to give a non-trivial data structure for
HCC queries on a general DAG, however, our goal is to reduce the space complexity rather than the
preprocessing time.
Surprisingly, this is possible and in fact we can achieve a substantial improvement in space complexity.

\begin{theorem}
    It is possible to solve the HCC problem (unweighted/weighted) on $O(n)$
    points in $\R$ on a DAG $\G$ of size $n$ using $\tilde{O}(n^2)$ preprocessing
    time, $\tilde{O}(n^{3/2})$ space, $O(\log(n))$ query time for unweighted and $O(\log\log(n))$ query time for weighted.
\end{theorem}
\begin{proof}
    We start by remarking that we cannot hope to reduce the preprocessing time, as shown by our conditional lower bound, 
    however and rather
    surprisingly, we show that the space can be reduced to $\tilde{O}(n^{3/2})$.

    Assume the input coordinates have been reduced to rank space (i.e., between 1 and $n$). 
    Let $I_q=(i,j)$ be the query interval. 
    Observe that we can afford to store the answer explicitly when $|I_q|\leq \sqrt{n}$ (i.e., ``short'' queries) since
    the number of such queries is $O(n^{3/2})$.
    This allows us to answer such queries in constant time.
    The subtle challenge, however, is to do it within $\tilde{O}(n^2)$ preprocessing time as the obvious
    solution could take much longer. 

    We start by computing the transitive closure $\G^c$ of $\G$, which takes
    $\tilde{O}(n^2)$ time. 
    For a point $p_i$ denote by $c_i$ its color in $\G$ and let $d_i$ be the number of parents of 
    $p_i$ in the transitive closure $\G^c$.
    We create $d_i$ copies of the point $p_i$ at the same position as $p_i$ and assign each a unique parent of $p_i$ as 
    color. 
    The end result will be a set of $O(n^2)$ points such that every point has a unique color and such that computing the number
    of colors in an interval $I=[I_1,I_2]$ will yield the answer to the hierarchical query with the same interval.
    This essentially gives a  ``flat'' representation of the hierarchical color structure, and 
    consequently, using the existing solutions for CRC queries, we can compute the answer to all
    the short queries in $\tilde{O}(n^{3/2})$ time and store the results in a table.
    Thus, short queries can be answered in constant time, using $\tilde{O}(n^{3/2})$ space and $\tilde{O}(n^2)$ preprocessing time. 

    It remains to show how to deal with the long queries.
    Note that we can repeat the above process for all the queries to obtain a ``flat representation'' 
    but doing so will yield a $\tilde{O}(n^2)$ space complexity. 
    The key idea here is that we can ``compress'' the flat representation  down to $O(n^{3/2})$ space, as follows.  
    Keep in mind that the compressed representation only needs to deal with
    queries $I_q$ such that $|I_q|\geq \sqrt{n}$.
    Partition the set of original $n$ points into $2\sqrt{n}$ subsets of $\sqrt{n}/2$ consecutive points. 
    For every subset $P_i$ and every color class $c$ (in the flat representation), delete all the points of color
    $c$ in $P_i$ except for the  points in the smallest and the largest position.
    This will leave at most two points of color $c$ in each subset and thus there will be at most $O(\sqrt{n})$ points of
    color $c$ in all subsets.
    Over $n$ colors this yields $O(n^{3/2})$ points. 
    We store them in a data structure for (weighted or unweighted) CRC queries and this will take
    $\tilde{O}(n^{3/2})$ space.

    The claim is that the compressed representation answers queries correctly.
    Consider a query interval $I_q$ of size at least $\sqrt{n}$ and assume to the contrary that there used to be a
    point $p$ of color $c$ inside $I_q$ in a subset $P_i$  but $p$ got deleted during the compression step. 
    However,  we have kept the rightmost point, $p_1$, and the leftmost point, $p_2$, of color $c$ 
    inside $P_i$. 
    Now, observe that since $P_i$ contains at most $\sqrt{n}/2$ points, either its rightmost point 
    or its leftmost point (or both) must be inside $I_q$. As a result, either $p_1$ or $p_2$ must be inside $I_q$, a contradiction.
    The preprocessing time here is also trivially $\tilde{O}(n^{2})$.
\end{proof}

\section{Sub-category range counting and equivalence}
Here, we will focus on SCRC queries.

\subsection{Equivalences}
We show that when the catalog tree $\G$ is a
path, then the SCRC problem in 1D is equivalent to a number of well-studied
problems, as follows.

\begin{theorem}\label{thm:scrceq}
  The following problems are equivalent:
  (i) SCRC when $\G$ is a single path on a one-dimensional input $P$, 
  (ii)   3-sided distinct coordinate counting (for a planar point set $P$),
  (iii)   3-sided color counting (for a planar point set $P$),
  (iv)  3D dominance color counting (for a 3D point set $P$), and finally
  (v)   3D dominance counting (for a 3D point set $P$).
\end{theorem}

To prove this, we first show the following lemma:

\begin{lemma}
  When $\G$ is a path, the SCRC problem on a one-dimensional input $P$ is equivalent to the 3-sided distinct coordinate counting problem. 
\end{lemma}

\begin{proof}
Here, we show that when $\G$ is a path, the SCRC Problem reduces 
to the 3-sided distinct coordinate counting problem.

To do that, first consider a 1D input point set for the SCRC problem.
Map the input point $x_i$ with color $c_i \in \G$, to the point $(x_i, h(c_i))$, where
$h(c_i)$ is the number of nodes below $c_i$ on the path $\G$.
Store the resulting point set in a data structure for the distinct coordinate counting
queries. 
Then, given a query interval $[a,b]$ and query node $c_q$, we create the
3-sided query range $\range = [a,b]\times (-\infty,h(c_q)]$. 
Observe that if a point $(x_i, h(c_q))$ lies inside $\range$, it implies that 
$c_i$ is below $c_q$ and that $x_i \in [a,b]$. 
Thus, the number of distinct $Y$-coordinates inside $\range$ is precisely the answer to the
SCRC problem.

To show the converse reduction, consider
a 2D point set $P$ for the 3-sided distinct coordinate counting. 
We create a node $v(y)$ for each distinct $Y$-coodinate $y$ in the input set,
meaning, $\G$ is a path that has as many vertices as the
number of distinct $Y$-coordinates in $P$. 

For each point $p_i=(x_i, y_i)$, we create a 1D point with x-coordinate $x_i$ and color 
$v(y_i)$. 
Consider a query range $\range=[a,b]\times
(-\infty, c]$ and let $y$ be the predecessor of $c$ among the $Y$-coordinates of $P$. 
We create the 1D range $[a,b]$ and query the node $v(y)$. 
It is straightforward to verify that the answer to the SCRC is exactly the number of
distinct $Y$-coordinates inside $\range$.
\end{proof}

The above lemma shows that (i) and (ii) are equivalent. 
Next, we observe that (ii), (iii), and (iv) all reduces to (v):
(iii) is a generalization of (ii), 
the reduction from (iii) to (iv) is standard
by mapping a 2D input point $(x_i,y_i)$ to the 3D point $(-x_i,y_i,x_i)$ and the 3-sided
query range $\range = [q_\ell, q_r] \times (-\infty, q_t]$ to the 3D dominance
range $(-\infty,-q_\ell]\times (-\infty,q_t] \times   (-\infty,q_x]$.
The reduction from (iv) to (v) was shown 
by Saladi~\cite{Saladi_2021}.
Thus, the only remaining piece of the puzzle is to show a reduction from (v) to (ii). 
We do this next.

\begin{theorem}
3D dominance counting can be solved by 3-sided distinct coordinate counting.
\end{theorem}
\begin{proof}
    Consider an instance of 3D dominance counting.
    First, we  reduce the instance to rank space which means we can assume that query coordinates
    are integers and the input coordinates are \textit{distinct} integers between 1 and $n$. 
    For an input point $p_i = (x_i,y_i,z_i)$, we create four points $p_i^1 = (-x_i,z_i),$ $p_i^2 = (y_i, z_i),$ $p_i^3=(-x_i, z_i)$ and $p_i^4 = (y_i, z_i+0.5)$. 
    Next, we create two 2D 3-sided distinct $Y$-coordinate counting structures. 
    In the first structure we put $p_i^1$ and $p_i^2$, and in the second structure we put $p_i^3$ and $p_i^4$. Given a query point $p=(x,y,z)$, we transform it into the 3-sided range
    $r_p = [-x,y]\times (\infty, z+0.5]$.

    We claim the number of dominated points of $p$ is the difference between the outputs of the two data structures on $r_p$. 
    If $p$ dominates $p_i$ then $x_i \leq x$, $y_i\leq y$ and $z_i\leq z$, hence $-x\leq -x_i$ and $z_i+0.5\leq z + 0.5$. This means that all four point are inside $r_p$. The first data structure counts $1$, since the two points share the second coordinate and the second data structure counts $2$ since they have distinct second coordinates. 

    The key observation is that if $p$ does not dominate $p_i$ at least one of $(-x_i, z_i)$ or $(y_i, z_i)$ and $(y_i, z_i+0.5)$ will not be in the range, 
    which consequently implies hence the difference of the outputs will be zero (regarding $p_i$). 
    The reduction is clearly linear.
\end{proof}

\begin{corollary}
    SCRC problem in 1D on category tree can be solved using $O(n\log^2 n/\log\log n)$ space and with the optimal query time of $O(\log n/\log\log n)$.
\end{corollary}
\begin{proof}
    We split the tree into its heavy path decomposition. 
    Once again, we need to look deeper into the details of the
    heavy-path decomposition.
    Start from the root  of $\G$ and follow a path $\pi$ to a leaf of $\G$  where at every
    node $u$ of $\G$, we always descend to a child of $u$ that has the largest subtree. 
    For every node $v \in \pi$, consider the subset $P_v \subset P$ that have a color from the set $\subcat(v)$.
    Build another instance of SCRC problem where the category graph is set to $\pi$, and a point
    $p \in P_v$ is assigned color $v$.
    In this instance of SCRC, the category graph is a path and by Theorem~\ref{thm:scrceq} it is equivalence to 3D dominance counting and thus
    it can be solved with $O(n \log n / \log\log n)$ space and with $O(\log n / \log\log n)$ query time~\cite{Jaja.ISAAC04}. 
    Observe that if for the query pair $(I,v_q)$, consisting of an interval $I$
    and a node $v_q \in \G$, we have $v_q \in \pi$, then the query can readily be answered. 
    Otherwise, $v_q$ must not be on the central path $\pi$.
    To handle such queries, we simply recurse on every tree that remains after removing $\pi$. 

    By the heavy path decomposition, the depth of the recursion is $O(\log n)$. 
    Furthermore, the paths obtained at the depth $i$ of the heavy path decomposition are independent and thus
    in total they contain $O(n)$ points. 
    Over all the $O(\log n)$ levels, it blows up the space by a factor of $O(\log n)$. 
    Note that to answer a query $(I,v_q)$, we need to find the level $i$ of the heavy path decomposition and a path
    $\pi_i$ that contains $v_q$.
    However, this can simply be answered by placing a pointer from $v_q$ to the appropriate data structure. 
    Then, after finding $\pi_i$, the query can be answered using a single dominance range counting query in 
    $O(\log n/\log\log n)$ time. 
\end{proof}
The query time of the above data structure is also optimal which follows from combining our reductions
with previously known lower bounds~\cite{patrascu08structures}.

\subsection{A Conditional Lower Bound for SCRC}
For SCRC where the underlying category graph is a DAG there is a trivial upper bound of $O(n^2)$ space and $O(\log n)$ query time:
for each node $v_i$ in the DAG store a 1D CRC structure on $\subcat(v_i)$. 
To answer a query $(I,v_q)$, we simply query the CRC structure on $v_q$ with the interval $I$. 
The conditional lower bound on HCC can easily be extended to SCRC. 
\begin{restatable}{corollary}{ovSCRC}
Assuming the Orthogonal Vectors conjecture, any solution to the 1D sub-category range counting problem on a DAG using $P(n)$ preprocessing time and $Q(n)$ query time, must obey $P(n)+nQ(n)\geq n^{2-o(1)}$. 
\end{restatable}

\begin{proof}
We pick the same underlying graph as in~\cref{ov_HCC} and the point set $P$ of $\eta$ points $p_i$ such that the color of $p_i$ is equivalent to node $b_i$. We query the SCRC $\eta$ times, each time using a different $a_i$ as our query node and an interval which spans all points of $P$. If for a query node $a_i$ the output is less than $\eta$ we know that $a_i$ is orthogonal to some $b_j$, since they do not share a coordinate where they both have a one.
\end{proof}


\bibliographystyle{abbrv}
\bibliography{ref}

\appendix
\section{HCC is at Least as Hard as Color Counting}\label{app:ashard}
Consider an instance of regular colored counting given by 
a point set $P$ in $\mathbb{R}^d$ where each point is assigned a color from a set $C$.
We build two hierarchical color counting data structures. 
In the first data structure, all colors in $C$ are represented by leaf nodes in a balanced binary tree $T_1$;
for simplicity we assume $|C|$ is a power of two; otherwise, we add dummy colors to $C$.
In the second data structure, for every pair of sibling leaves $c$ and $c'$, we ``collapse them'' into their parent $c_p$,
to obtain a second balanced binary tree $T_2$, meaning, 
any point that had color $c$ or $c'$ will receive $c_p$ as its color. 
The resulting colored point set will be stored in the second hierarchical color counting data structure. 

Given a query range $\range$ for the regular colored counting  problem, 
we query both data structures with $\range$ and report their difference.
We claim it will be the correct answer to the regular colored counting query.

To see this, we can consider two sibling leaf colors $c$ and $c'$ and their parent $c_p$ in $T_1$.
Case one is when none of them is in the query range $\range$. 
In this case, neither data structure will count anything and thus
their difference will also not count either color. 
The second case is when exactly one of them, say $c$ is in the query. 
In this case, the first data structure counts the leaf $c$ and the path connecting $c_p$ to the root of $T_1$
where as the second counts only the latter and thus the difference counts $c$ exactly once.
Lastly if both leaves are in the query, the first
data structure counts two more leafs when compared to the second data structure (the
two leaf nodes in the query) and we again get the correct output. The three
cases are summarized in Figure \ref{normal_to_tree}.

\begin{figure}[h!]
  \centering{\includegraphics[scale=1.2]{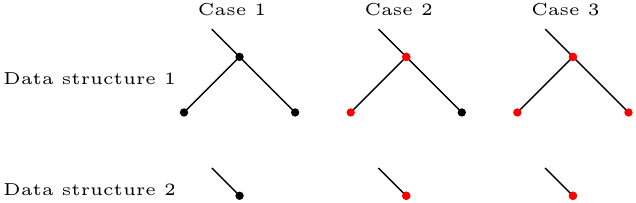}}
\caption{Each case, where the red dots corresponds to points counted}
\label{normal_to_tree}
\end{figure}

\end{document}